\documentclass[fontsize=11pt,DIV=13,paper=letter]{scrartcl}
\usepackage[pass,letterpaper]{geometry}

\usepackage{amssymb,amsmath,amsfonts,amsthm,enumerate,array,multirow}

\newtheorem{theorem}{Theorem}
\newtheorem{proposition}[theorem]{Proposition}

\newtheorem{lemma}[theorem]{Lemma}
\newtheorem{corollary}[theorem]{Corollary}

\theoremstyle{definition}
\newtheorem{definition}[theorem]{Definition}

\newcommand{\sub}{\subseteq}
\newcommand{\la}{\langle}
\newcommand{\ra}{\rangle}
\newcommand{\A}{\mathcal{A}}
\newcommand{\Z}{\mathbb{Z}}

\newcommand{\F}{\mathcal{F}}
\newcommand{\s}{\mathcal{S}}
\newcommand{\LL}{\mathcal{L}}
\newcommand{\SL}{\mathrm{SL}(2,\Z)}
\newcommand{\GL}{\mathrm{GL}(2,\Z)}
\renewcommand{\phi}{\varphi}
\renewcommand{\epsilon}{\varepsilon}

\begin{document}

\title{\Large Decidability of the Membership Problem for $2\times 2$ integer matrices\thanks{This research was supported by EPSRC grant EP/M00077X/1.}}

\author{\large Igor Potapov\thanks{Department of Computer Science, University of Liverpool. Email: {\tt potapov@liverpool.ac.uk}}
\and \large Pavel Semukhin\thanks{Department of Computer Science, University of Liverpool. Email: {\tt semukhin@liverpool.ac.uk}}}
\date{}

\maketitle

\begin{abstract}
The main result of this paper is the decidability of the membership problem for $2\times 2$ nonsingular integer matrices. Namely, we will construct the first algorithm that for any nonsingular $2\times 2$ integer matrices $M_1,\dots,M_n$ and $M$ decides whether $M$ belongs to the semigroup generated by $\{M_1,\dots,M_n\}$.

Our algorithm relies on a translation of the numerical problem on matrices into combinatorial problems on words. It also makes use of some algebraical properties of well-known subgroups of $\GL$ and various new techniques and constructions that help to limit an infinite number of possibilities by reducing them to the membership problem for regular languages.
\end{abstract}

\section{Introduction}
Matrices and matrix products play a crucial role in a representation and analysis of various computational processes, 
i.e., linear recurrent sequences \cite{tHaHaHiKa05a,OW_ICALP2015-1,OW_ICALP2015-2}, arithmetic circuits \cite{GOW_STACS2015}, 
hybrid and dynamical systems \cite{OSW2015, BP2008},
probabilistic and quantum automata \cite{Blondel2005}, stochastic
games, broadcast protocols \cite{EFM1999},  optical systems \cite{GB94}, etc. 
Unfortunately, many simply formulated and elementary problems for matrices are inherently difficult to solve
even in dimension two, and most of these problems 
become undecidable in general starting from dimension three or four.
One of such hard questions is the {\sl Membership problem} in matrix semigroups:\\

\noindent
{\bf Membership problem:} Given a finite set of $m \times m$ matrices $F=\{M_1,M_2, \ldots ,M_n\}$  and a matrix $M$.
Determine if there exist an integer $k \ge 1$ and $i_1,i_2,\ldots ,i_k \in \{1,\ldots ,n\}$ such that 
$M_{i_1} \cdot M_{i_2}  \cdots  M_{i_k}=M$.
In other words, determine whether a matrix $M$ belongs to the semigroup generated by $F$.\\

In this paper we solve an open problem by showing that  
the membership is decidable
for the semigroups of $2 \times 2$ nonsingular matrices over integers.
%
The membership problem was intensively studied since 1947 when A.Markov showed that 
this problem is undecidable for matrices in $\Z^{6 \times 6}$ even for a specific fixed set $F$ \cite{Markov}. 
Later, M. Paterson in 1970 showed that a special case of the membership problem when $M$ is equal to a zero matrix (known as {\sl Mortality problem})
is undecidable for matrices in $\Z^{3 \times 3}$. The decidability status of another special case of 
the membership problem --- the {\sl Identity problem} (i.e., when $M=I$, the identity matrix) --- was unknown for a long time and
was only recently shown to be undecidable for integer matrices starting from dimension four \cite{Identity}, see also the solution to Problem 10.3 in \cite{solution10-3}. The undecidability of the
identity problem means that the {\sl Group problem} (of whether a matrix semigroup over integers forms a group) is undecidable 
starting from dimension four. A more recent survey of undecidable problems can be found in~\cite{CassaigneHHN14}.

The undecidability proofs in matrix semigroups are mainly based on  various techniques and methods for embedding
universal computations into matrix products. The case of dimension two is the most intriguing
since there is some evidence that if these problems are undecidable, then this cannot be proved using 
any previously known constructions. In particular, there is no injective semigroup morphism from pairs of words over any finite
alphabet (with at least two elements) into complex $2 \times 2$ matrices \cite{CHK99},
which means that the coding of independent pairs of words in $2 \times 2$ complex matrices is impossible and the exact encoding of
the Post Correspondence Problem or a computation of the Turing Machine cannot be used directly for proving undecidability
in $2 \times 2$ matrix semigroups over $\mathbb{Z}$, $\mathbb{Q}$ or $\mathbb{C}$.
The only undecidability in the case of $2 \times 2$ matrices has been shown so far is the membership, 
freeness and vector reachability problems over quaternions \cite{BP_IC2008} or more precisely in the case of diagonal matrices over quaternions, 
which are simply double quaternions.
 

The problems for semigroups are rather hard, but there was a steady progress on decidable fragments over the last few decades.
First, both membership and vector reachability problems were shown to be decidable in polynomial time
for a semigroup generated by a single  $m \times m$ matrix (known as the {\sl Orbit problem}) by Kannan and Lipton \cite{KL86} in 1986.
Later, in 1996 this decidability result was extended to a more general case of  
commutative matrices \cite{Babai}.
The generalization of this result for a special class of non-commutative matrices
(a class of row-monomial matrices over a commutative semigroup satisfying some natural effectiveness conditions)
was shown in 2004 in~\cite{LP2004}.
Even now we still have long standing open problems for matrix semigroups generated by a single matrix,
see, for example, the {\sl Skolem Problem} about reaching zero in a linear recurrence sequence (LRS), which in matrix 
form is a question of whether any power of a given integer matrix $A$ has zero in the right upper corner \cite{STOC2013,COW_JACM}. It was recently shown
that the decidability of either Positivity or Ultimate Positivity
for integer LRS of order 6 would entail some major breakthroughs
in analytic number theory. The decidability of each of
these problems, whether for integer, rational, or algebraic
linear recurrence sequences, is open, although partial results
are known \cite{GOW_STACS2015,OSW2015,OW_ICALP2015-1,OW_ICALP2015-2}.

Due to a severe lack of methods and techniques the status of decision problems for $2 \times 2$ matrices (like membership, vector reachability, freeness) is remaining to be a long standing open problem. More recently, a new approach of translating numerical problems of $2 \times 2$ integer matrices into variety of combinatorial and computational problems on words over group alphabet and studying their transformations as specific rewriting systems have  led to a few results on decidability and complexity for some subclasses. 
In particular, this approach was successfully applied to proving the decidability of the membership problem for semigroups from $\GL$ \cite{CK2005} in 2005,  designing the polynomial time algorithm for the membership problem for the modular group \cite{Gurevich2007} in 2007,  showing NP-hardness for most of the reachability problems in dimension two \cite{BP2012,BHP2012} in 2012,
and showing decidability of the vector/scalar reachability problems in $\SL$ \cite{PS15} in 2015.

The main ingredient of the translation into combinatorial problems on words is the well-known result that the groups $\SL$ and $\GL$ are finitely generated. For example, $\SL$ can be generated by a pair of matrices: 
\begin{center}
$S = \begin{bmatrix}
0 & -1 \\ 1 & 0
\end{bmatrix}
$ and
$R = \begin{bmatrix}
0 & -1 \\ 1 & 1
\end{bmatrix}
$
with the following relations: $S^4=I$, $R^6=I$ and $S^2=R^3$.
\end{center}
Hence we can represent a matrix $M\in \SL$ as a word in the alphabet $\{S,R\}$.

In \cite{CK2005} both the {\sl Identity} and the {\sl Group} problems are shown to be decidable in  $\Z^{2 \times 2}$. 
Moreover, it was also claimed more generally
that it is decidable whether or not a given nonsingular matrix belongs to a given finitely generated semigroup over integers.
Unfortunately, it appears that the proof of this more general claim (i.e., when we consider matrices with determinants different from~$\pm 1$) has a significant gap,
and it only works for a small number of special cases. Namely, after translating the membership from $\GL$ to $\SL$, the authors describe a very short reduction 
from the membership problems in $\Z^{2 \times 2}$ to the one in $\SL$ using some incorrect assumptions.
For instance, it was assumed that if $X$ is an integer matrix with determinant one and $Z$ is a nonsingular integer matrix, then 
there exists an integer matrix $Y$ satisfying the following equation  $ZX=YZ$. However, this is not true and here is a simple counter example.
Let $Z = \begin{bmatrix}
1 & 0 \\ 0 & 2
\end{bmatrix}
$ and $X = \begin{bmatrix}
0 & -1 \\ 1 & 0
\end{bmatrix}
$, then from $ZX=YZ$ it follows that 
$Y = Z X Z^{-1} =$
$
\begin{bmatrix}
1 & 0 \\ 0 & 2
\end{bmatrix}
 \times
\begin{bmatrix}
0 & -1 \\ 1 & 0
\end{bmatrix}
 \times
\begin{bmatrix}
1 & 0 \\ 0 & \frac{1}{2}
\end{bmatrix}
=
\begin{bmatrix}
0 & -\frac{1}{2} \\ 2 & 0
\end{bmatrix}
$. So $Y$ has fractional coefficients, and if the matrices $X$ and $Z$ were in the generating set, then the argument from \cite{CK2005} would not work.

The main result of this paper is that the membership problem is decidable
for the semigroups of $2 \times 2$ nonsingular integer matrices. Our proof provides an algorithm for solving this problem, which is
based on the translation of the numerical problem on matrices into combinatorial problems on words and regular languages. We will also makes use of some well-known algebraical results like the uniqueness of the Smith normal form of a matrix and a fact that certain subgroups of $\GL$ have finite index.

\section{Preliminaries}

The semigroup of $2\times 2$ integer matrices is denoted by $\Z^{2\times 2}$.
We use $\SL$ to denote the special linear group of $2\times 2$ matrices with integer coefficients, i.e., $\SL=\{M\in \Z^{2\times 2} : \det(M)=1\}$ and $\GL$ to denote the general linear group, i.e., $\GL=\{M\in \Z^{2\times 2} : \det(M)=\pm 1\}$.

A matrix is called \emph{nonsingular} if its determinant is not equal to zero.

If $F$ is a finite collection of matrices from $\Z^{2\times 2}$, then $\la F\ra$ denotes the semigroup generated by $F$ (including the identity matrix), that is, $M\in \la F\ra$ if and only if $M=I$ or there are matrices $M_1,\dots,M_n\in F$ such that $M=M_1\cdots M_n$.

\section{Main result}

The main result of our paper is presented in Theorem \ref{thm:mem} which states that membership problem in dimension two is decidable.

\begin{theorem} \label{thm:mem}
There is an algorithm that decides for a given finite collection $F$ of nonsingular matrices from $\Z^{2\times 2}$ and a matrix $M\in \Z^{2\times 2}$ whether $M\in \la F\ra$.
\end{theorem}

{\it Proof sketch.} Let $\{M_1,\dots,M_n\}$ be all matrices from $F$ whose determinant is different from~$\pm 1$, and let $\s^{\pm 1}$ be the semigroup which is generated by all matrices from $F$ with determinant $\pm 1$, that is, $\s^{\pm 1} = \la\, F\cap \GL\,\ra$. Then it is not hard to see that $M\in \la F\ra$ if and only if $M\in \s^{\pm 1}$ or there is a sequence of indices $i_1,\dots,i_t\in \{1,\dots,n\}$ and matrices $A_1,\dots,A_{t+1}$ from $\s^{\pm 1}$ such that
\[
M=A_1M_{i_1}A_2M_{i_2}\cdots A_tM_{i_t}A_{t+1}.
\]

The key point of the proof is that the value of $t$ is bounded. Indeed, since $|\det(M_{i_s})|\geq 2$, for $s=1,\dots,t$, we have that $t\leq \log_2|\det(M)|$. So to decide whether or not $M\in \la F\ra$ we first need to check whether $M\in \s^{\pm 1}$. If $M\notin \s^{\pm 1}$, then we need to go through all sequences $i_1,\dots,i_t\in \{1,\dots,n\}$ of length up to $\log_2|\det(M)|$ and for every such sequence check whether there are matrices $A_1,\dots,A_{t+1}$ from $\s^{\pm 1}$ such that $M=A_1M_{i_1}A_2M_{i_2}\cdots A_tM_{i_t}A_{t+1}$. The rest of the paper is devoted to the proof that these problems are algorithmically decidable.

In Section \ref{GL} we describe an algorithm that decides whether $M\in \s^{\pm 1}$. In fact, in Proposition~\ref{MPReg} we prove a stronger statement that it is decidable whether $M\in \s$, where $\s$ is an arbitrary regular subset of $\GL$, that is, a subset which is defined by a finite automaton. The precise definition of this notion is given in Section \ref{GL}. We will also show there that any semigroup in $\GL$, and in particular $\s^{\pm 1}$, is a regular subset.

Proposition \ref{MPReg} provides an alternative proof for the decidability of the membership in $\GL$ presented in \cite{CK2005}. The difference of our approach is that we do not introduce new symbols in the alphabet, and we explicitly construct an automaton $\mathrm{Can}(\A)$ that accepts only canonical words. The construction of $\mathrm{Can}(\A)$ will be also used in the next steps of our algorithm.

In Section \ref{1mat} we provide a proof for the decidability of the second problem in the special case when $t=1$. Again, in Corollary \ref{cor:1diag} we prove a more general statement that for any two nonsingular matrices $M_1$ and $M_2$ from $\Z^{2\times 2}$ and regular subsets $\s_1$ and $\s_2$, it is decidable whether there are matrices $A_1\in \s_1$ and $A_2\in \s_2$ such that $A_1M_1A_2=M_2$.

Finally, in Section \ref{Gen} we describe an algorithm for the general case. Namely, in Theorem \ref{thm:diag} we will prove that for any nonsingular matrices $M_1,\dots,M_t$ from $\Z^{2\times 2}$ and for any regular subsets $\s_1,\dots,\s_t$ of~$\GL$, it is decidable whether there are matrices $A_1\in \s_1,\dots,A_t\in \s_t$ such that $A_1M_1\cdots A_{t-1}M_{t-1}A_t=M_t$.

\qed

{\it Remark.} The complexity of our algorithm is in EXPSPACE. The exponential blow-up in memory usage happens when we translate matrices into words and construct a finite automaton for the semigroup $\s^{\pm 1}$ (see the paragraph before Corollary \ref{cor:GL} in Section \ref{GL}). The other steps of the algorithm require only polynomial space.
Furthermore, our algorithm can be extended to check the membership not only for semigroups in $\Z^{2\times 2}$ but for arbitrary regular subsets of nonsingular matrices from $\Z^{2\times 2}$.

\subsection{Decidability of the membership problem in $\GL$.}\label{GL}

We will use an encoding of matrices from $\GL$ by words in alphabet $\Sigma=\{X,N,S,R\}$. For this we define a mapping $\phi: \Sigma\to \GL$ as follows:
\[
\phi(X)=\begin{bmatrix} -1 & 0\\ 0 & -1\end{bmatrix}\!,\
\phi(N)=\begin{bmatrix} 1 & 0\\ 0 & -1\end{bmatrix}\!,\
\phi(S)=\begin{bmatrix} 0 & -1\\ 1 & 0\end{bmatrix}\!,\
\phi(R)=\begin{bmatrix} 0 & -1\\ 1 & 1\end{bmatrix}\!.
\]
We can extend $\phi$ to the morphism $\phi: \Sigma^* \to \GL$ in a natural way. It is a well-known fact that morphism $\phi$ is surjective, that is, for every $M\in \GL$ there is a word $w\in \Sigma^*$ such that $\phi(w)=M$.

\begin{definition}
We call two words $w_1$ and $w_2$ from $\Sigma^*$ \emph{equivalent}, denoted $w_1\sim w_2$, if $\phi(w_1)=\phi(w_2)$.

Two languages $L_1$ and $L_2$ in the alphabet $\Sigma$ are \emph{equivalent}, denoted $L_1\sim L_2$, if
\begin{enumerate}[(i)]
\item for each $w_1\in L_1$, there exists $w_2\in L_2$ such that $w_1\sim w_2$, and
\item for each $w_2\in L_2$, there exists $w_1\in L_1$ such that $w_2\sim w_1$.
\end{enumerate}
In other words, $L_1\sim L_2$ if and only if $\phi(L_1)=\phi(L_2)$.
Two finite automata $\A_1$ and $\A_2$ with alphabet $\Sigma$ are \emph{equivalent}, denoted $\A_1\sim \A_2$, if $L(\A_1)\sim L(\A_2)$.
\end{definition}

To simplify the notation we will often write $M=w$ instead of $M=\phi(w)$ when $M\in \GL$ and $w\in \Sigma^*$. Note that in this notation if $M=w_1$ and $M=w_2$, then we have $w_1\sim w_2$ but not necessarily $w_1=w_2$.

\begin{definition}
A subset $\s\sub \GL$ is called \emph{regular} or \emph{automatic} if there is a regular language $L$ in alphabet $\Sigma$ such that $\s=\phi(L)$.
\end{definition}

Throughout the paper we will use the following abbreviation: if $n$ is a positive integer and $V\in \Sigma$, then $V^n$ denotes a words of length $n$ which contains only letter $V$, and $V^0$ is assumed to be equal to the empty word.

\begin{definition}
A word $w\in \Sigma^*$ is called a \emph{canonical word} if it has the form
\[
w=N^\delta X^\gamma S^{\beta}R^{\alpha_0}SR^{\alpha_1}SR^{\alpha_2}\dots SR^{\alpha_{n-1}}SR^{\alpha_n},
\]
where $\beta,\delta,\gamma\in \{0,1\}$, $\alpha_0,\dots,\alpha_{n-1}\in \{1,2\}$, and $\alpha_n\in \{0,1,2\}$. In other words, $w$ is \emph{canonical} if it does not contain subwords $SS$ or $RRR$. Moreover, letter $N$ may appear only once in the first position, and letter $X$ may appear only once either in the first position or after $N$.
\end{definition}

We will make use of Corollary \ref{cor:can} below which states that every matrix from $\GL$ can be represented by a unique canonical word.

\begin{proposition}[\cite{LS,MKS,Ran}]\label{prop:can}
For every matrix $M\in \SL$, there is a unique canonical word $w$ such that $M=w$. Note that $w$ does not contain letter $N$ because $\phi(N)\notin \SL$.
\end{proposition}

\begin{corollary}\label{cor:can}
For every matrix $M\in \GL$, there is a unique canonical word $w$ such that $M=w$.
\end{corollary}

\begin{proof}
If $\det(A)=1$, that is, $M\in \SL$, then by Proposition \ref{prop:can} there is a unique canonical word $w$ such that $M=w$. If $\det(A)=-1$, then $N^{-1}M\in \SL$ and again by Proposition \ref{prop:can} there is a unique canonical word $w$ such that $N^{-1}M=w$ or $M=Nw$. Note that $Nw$ is also a canonical word since $w$ does not contain letter $N$.
\end{proof}

\begin{proposition}\label{MPReg}
There is an algorithm that for any regular subset $\s\sub \GL$ and a matrix $M\in \GL$ decides whether $M\in \s$.
\end{proposition}

\begin{proof}
Let $L$ be a regular language such that $\s=\phi(L)$, and let $\A$ be a finite automaton that recognizes $L$, that is, $L=L(\A)$. The words in $L$ do not have to be in canonical form. So, we will construct a new automaton $\mathrm{Can}(\A)$ whose language contains only canonical words and such that $\mathrm{Can}(\A)$ is equivalent to $\A$, that is, $\phi(L(\mathrm{Can}(\A)))=\phi(L(\A))=\s$. The construction of $\mathrm{Can}(\A)$ consists of a sequence of transformations that insert new paths and $\epsilon$-transitions into $\A$. The detailed description of this construction is given in Section \ref{Can} of the Appendix.

Using the automaton $\mathrm{Can}(\A)$ we can decide whether $M\in \s$. Indeed, by Corollary \ref{cor:can}, there is a unique canonical word $w$ that represents the matrix $M$, i.e., $M=\phi(w)$. Now we have the following equivalence: $M\in \s$ if and only if $w\in L(\mathrm{Can}(\A))$. Therefore, to decide whether $M\in \s$, we need to check whether $w$ is accepted by $\mathrm{Can}(\A)$.

\end{proof}

Note that any finitely generated semigroup $\la M_1,\dots,M_n\ra$ in $\GL$ is a regular subset. Indeed, let $w_1,\dots,w_n$ be canonical words that represent the matrices $M_1,\dots,M_n$, respectively, and consider a regular language $L={(w_1+\cdots+w_n)}^*$. Clearly $\phi(L)=\la M_1,\dots,M_n\ra$, and hence the semigroup $\la M_1,\dots,M_n\ra$ is regular.
So as a corollary from Proposition \ref{MPReg} we obtain the decidability of the membership problem for semigroups in $\GL$.

\begin{corollary}\label{cor:GL}
The membership problem for $\GL$ is decidable. That is, there is an algorithm that for a given finite collection of matrices $M_1,\dots,M_n$ and $M$ from $\GL$, decides whether $M\in \la M_1,\dots,M_n\ra$.
\end{corollary}

\subsection{Special case: $A_1M_1A_2=M_2$}\label{1mat}

In this section we show that for any two nonsingular matrices $M_1$ and $M_2$ from $\Z^{2\times 2}$ and regular subsets $\s_1$ and $\s_2$, it is decidable whether there exist matrices $A_1\in \s_1$ and $A_2\in \s_2$ such that $A_1M_1A_2=M_2$ (Corollary \ref{cor:1diag}). First, we prove this statement in the case when $M_1=M_2=D$, where $D$ is a diagonal matrix in the Smith normal form (Proposition \ref{prop:1diag}).

For the proof of this result we will use a few algebraical facts and results that are explained below. The most important of them is the following theorem about the Smith normal form of a matrix.

\begin{theorem}[Smith normal form \cite{KB79}] \label{SNF}
For any matrix $A\in \Z^{2\times 2}$, there are matrices $E,F$ from $\GL$ such that
\[
A=E\begin{bmatrix} t_1 & 0\\ 0 & t_2 \end{bmatrix}F
\]
for some $t_1,t_2\in \Z$ such that $t_1 \mid t_2$. The diagonal matrix $\begin{bmatrix} t_1 & 0\\ 0 & t_2 \end{bmatrix}$, which is unique up to the signs of $t_1$ and $t_2$, is called the \emph{Smith normal form} of $A$. Moreover, $E$, $F$, $t_1$, and $t_2$ can be computed in polynomial time.
\end{theorem}

\begin{definition}
If $H$ is a subgroup of $G$, then the sets $gH = \{gh\ :\ h\in H\}$ and $Hg = \{hg\ :\ h\in H\}$, for $g\in G$, are called the \emph{left} and \emph{right cosets} of $H$ in $G$, respectively. An element $g$ is called a \emph{representative} of the left coset $gH$ (respectively, of the right coset $Hg$).

The collection of left cosets or right cosets of $H$ form a disjoint partition of $G$. Moreover, the number of left cosets is equal to the number of right cosets, and this number is called the \emph{index} of $H$ in $G$, denoted $|G:H|$. 
\end{definition}

For every natural $n\geq 1$, let us define the following subgroups of $\GL$:
\[
\begin{split}
H(n) &= \left\{\begin{bmatrix} a_{11} & a_{12}\\ a_{21} & a_{22}\end{bmatrix}\in \GL\ :\ n \text{ divides } a_{21}\right\},\\ 
F(n) &= \left\{\begin{bmatrix} a_{11} & a_{12}\\ a_{21} & a_{22}\end{bmatrix}\in \GL\ :\ n \text{ divides } a_{12}\right\}.
\end{split}
\]

Let $A=\begin{bmatrix} a_{11} & a_{12}\\ a_{21} & a_{22} \end{bmatrix}$ be any matrix from $\GL$ and let $D=\begin{bmatrix} m & 0\\ 0 & mn \end{bmatrix}$ be a diagonal matrix in the Smith normal form, where $m,n\neq 0$. Then the conjugation of $A$ with $D$ is equal to
\[
A^D = D^{-1}AD = \begin{bmatrix} a_{11} & na_{12}\\ \frac{1}{n}a_{21} & a_{22} \end{bmatrix}\!.
\]

From this formula we see that if $A^D\in \GL$, then $n$ divides $a_{21}$. On the other hand, if $a_{21}$ is divisible by $n$, then $A^D$ is in $\GL$, and in fact in $F(n)$. Thus we have the following criterion.
\begin{proposition}\label{prop:conj}
Suppose $A$ is in $\GL$ and $D$ is a diagonal matrix of the above form, then $A^D\in \GL$ if and only if $A\in H(n)$. Moreover, if $A\in H(n)$, then $A^D\in F(n)$.
\end{proposition}

\begin{theorem}\label{thm:ind}
The subgroups $H(n)$ and $F(n)$ have finite index in $\GL$. Furthermore, there is an algorithm that for a given $n$ computes representatives of the left and right cosets of $H(n)$ and $F(n)$ in $\GL$.
\end{theorem}

\begin{proof}
We will only show how to compute representatives of the left cosets of $H(n)$ because the other cases are similar. For each pair of indices $i,j$ such that $0\leq i,j\leq n-1$, let us define a matrix $W_{i,j}$ as follows. Let $W_{i,0}$ be the identity matrix for $i=0,\dots,n-1$. If $j>0$, then consider $d=\gcd(i,j)$ and let $i_0$ and $j_0$ be such that $i=i_0d$ and $j=j_0d$. Since $i_0,j_0$ are relatively prime, there exist integers $u$ and $v$ such that $ui_0+vj_0=1$. Hence if we let $W_{i,j}=\begin{bmatrix} u & v\\ -j_0 & i_0 \end{bmatrix}$, then $W_{i,j}$ belongs to $\GL$.

Now consider an arbitrary matrix $A=\begin{bmatrix} a_{11} & a_{12}\\ a_{21} & a_{22} \end{bmatrix}$ from $\GL$. Let $a_{11}=i+nk$ and $a_{21}=j+nl$, where $0\leq i,j\leq n-1$. We will show that $W_{i,j}A\in H(n)$. If $j=0$, then $a_{21}=nl$ is divisible by $n$, and hence $A\in H(n)$. Since we defined $W_{i,0}$ to be the identity matrix, it follows that $W_{i,0}A = A\in H(n)$. If $j>0$, then let $d=\gcd(i,j)$ and let $i_0,j_0$ be such that $i=i_0d$ and $j=j_0d$. In this case
\[
W_{i,j}A = \begin{bmatrix} u & v\\ -j_0 & i_0 \end{bmatrix}\begin{bmatrix} di_0+nk & a_{12}\\ dj_0+nl & a_{22} \end{bmatrix}\!,
\]
and the lower left corner of $W_{i,j}A$ is equal to $-j_0di_0-j_0nk+i_0dj_0+i_0nl = n(-j_0k+i_0l)$, which is divisible by $n$. Thus $W_{i,j}A\in H(n)$.

So we showed that for any matrix $A\in \GL$ there is a pair $i,j$ such that $W_{i,j}A\in H(n)$ or, equivalently, $A\in W^{-1}_{i,j}H(n)$. Therefore, the collection $\{W^{-1}_{i,j}H(n)\ :\ 0\leq i,j\leq n-1\}$ contains all left cosets of $H(n)$ in $\GL$. In particular, the index of $H(n)$ in $\GL$ is bounded by $n^2$.

Note that some of the cosets in $\{W^{-1}_{i,j}H(n)\ :\ 0\leq i,j\leq n-1\}$ may be equal to each other. In fact, two cosets $W^{-1}_{i_1,j_1}H(n)$ and $W^{-1}_{i_2,j_2}H(n)$ are equal if and only if $W^{}_{i_1,j_1}W^{-1}_{i_2,j_2}\in H(n)$. Since the domain of the subgroup $H(n)$ is a computable set, the equality of two cosets is a decidable property. Therefore, we can algorithmically choose a collection of pairwise nonequivalent representatives of the left cosets of $H(n)$ in $\GL$.

\end{proof}

\begin{lemma}
\label{lem:regsub}
Let $L_{H(n)}$ and $L_{F(n)}$ be the languages that correspond to the subgroups $H(n)$ and $F(n)$, respectively, that is, $L_{H(n)} = \{w\in \Sigma^* : \phi(w)\in H(n)\}$ and $L_{F(n)} = \{w\in \Sigma^* : \phi(w)\in F(n)\}$. Then $L_{H(n)}$ and $L_{F(n)}$ are regular languages.
\end{lemma}

\begin{proof}
We will show that $L_{H(n)}$ is regular by constructing an automaton $\A_{H(n)}$ that recognizes it. The proof for $L_{F(n)}$ is similar.

Let $U_0,U_1,\dots,U_k$ be pairwise nonequivalent representatives of the right cosets of $H(n)$ in $\GL$, which can be computed by Theorem \ref{thm:ind}. We will assume that $U_0=I$ and hence $H(n)U_0=H(n)$. The automaton $\A_{H(n)}$ will have $k$ states $u_0,u_1,\dots,u_k$, where $u_0$ is the only initial and the only final state of $\A_{H(n)}$. The transitions of $\A_{H(n)}$ are defined as follows: there is a transition from $u_i$ to $u_j$ labelled by $\sigma\in \Sigma$ if and only if the element $U_i\phi(\sigma)$ belongs to the coset $H(n)U_j$. Note that since for every $i$ and $\sigma$ there is exactly one $j$ such that $U_i\phi(\sigma)\in H(n)U_j$, the automaton $\A_{H(n)}$ is deterministic.

We now show that the language of $\A_{H(n)}$ is equal to $L_{H(n)}$. Take any word $w=\sigma_1\sigma_2\dots\sigma_t \in \Sigma^*$ and consider a run $\rho = u_{i_0}u_{i_1}\dots u_{i_t}$ of $\A_{H(n)}$ on $w$. Note that $i_0=0$, and $u_{i_0}=u_0$ is the initial state. Since $\A_{H(n)}$ has transitions $u_{i_{s-1}} \xrightarrow{\sigma_s} u_{i_s}$, for $s=1,\dots,t$, we have that $U_{i_{s-1}}\phi(\sigma_s)\in H(n)U_{i_s}$ and hence $U^{}_{i_{s-1}}\phi(\sigma_s)U_{i_s}^{-1}\in H(n)$. Since $U_{i_0}=U_0=I$, we can rewrite $\phi(w)=\phi(\sigma_1)\phi(\sigma_2)\dots\phi(\sigma_t)$ as
\[
\phi(w)=(U^{}_{i_0}\phi(\sigma_1)U_{i_1}^{-1})(U^{}_{i_1}\phi(\sigma_2)U_{i_2}^{-1})\cdots (U^{}_{i_{t-1}}\phi(\sigma_t)U_{i_t}^{-1})U^{}_{i_t}.
\]
If $u_{i_t}=u_0$, that is, if $w$ is accepted by $\A_{H(n)}$, then $i_t=0$ and $U_{i_t} = U_0=I\in H(n)$. This implies that $\phi(w)\in H(n)$ because for all $s=1,\dots,t$ we have $U^{}_{i_{s-1}}\phi(\sigma_s)U_{i_s}^{-1}\in H(n)$. On the other hand, if $\phi(w)\in H(n)$, then it must be that $U_{i_t}\in H(n)$, which can only happen if $i_t=0$ and hence $u_{i_t}=u_0$. This means that $w$ is accepted by $\A_{H(n)}$. Therefore, we proved that $L(\A_{H(n)}) = L_{H(n)}$.

\end{proof}

Now for any automaton $\A$ with alphabet $\Sigma$ we construct two automata $\mathrm{Inv}(\A)$ and $\F_D(\A)$, where $D$ is a diagonal matrix in the Smith normal form. The automaton $\mathrm{Inv}(\A)$ recognizes inverses to the words from $L(\A)$, that is:
\begin{enumerate}[(1)]
\item For every $w\in L(\A)$, there exists $w'\in L(\mathrm{Inv}(\A))$ such that $\phi(w')=\phi(w)^{-1}$.
\item For every $w'\in L(\mathrm{Inv}(\A))$, there exists $w\in L(\A)$ such that $\phi(w)=\phi(w')^{-1}$.
\end{enumerate}
In other words, for any matrix $A\in \GL$, $A\in \phi(L(\A))$ if and only if $A^{-1}\in \phi(L(\mathrm{Inv}(\A)))$.

{\bf Construction of the automaton $\mathrm{Inv}(\A)$.}
We will make use of the following equivalences, which are easy to check: $X^{-1}\sim X$, $N^{-1}\sim N$, $S^{-1}\sim S^3$, and $R^{-1}\sim R^5$. Informally speaking, to construct $\mathrm{Inv}(\A)$ we want to reverse the transitions in $\A$ and replace the labels by their inverses. More formally, $\mathrm{Inv}(\A)$ will have the same states as $\A$ plus some newly added states as explained below. The initial states of $\mathrm{Inv}(\A)$ are the final states of $\A$, and the final states of $\mathrm{Inv}(\A)$ are the initial states of $\A$. For every transitions of the form $q\xrightarrow{X} q'$ and $q\xrightarrow{N} q'$ in $\A$ we add the transitions $q'\xrightarrow{X} q$ and $q'\xrightarrow{N} q$ to $\mathrm{Inv}(\A)$, respectively. Furthermore, for every transitions of the form $q\xrightarrow{S} q'$ and $q\xrightarrow{R} q'$ in $\A$ we add the paths $q'\xrightarrow{S} p_1 \xrightarrow{S} p_2 \xrightarrow{S} q$ and $q'\xrightarrow{R} p_3 \xrightarrow{R} p_4 \xrightarrow{R} p_5 \xrightarrow{R} p_6 \xrightarrow{R} q$ to $\mathrm{Inv}(\A)$, respectively, where $p_1,p_2,\dots,p_6$ are newly added states. It is not hard to verify that $\mathrm{Inv}(\A)$ has the desired properties.

\medskip
The purpose of the automaton $\F_D(\A)$ is to recognize conjugations of the words from $L(\A)$ with matrix $D$. To explain formally what this means, let $D=\begin{bmatrix} m & 0\\ 0 & mn \end{bmatrix}$ be a diagonal matrix in the Smith normal form, where $m,n\neq 0$. Recall that by Proposition \ref{prop:conj}, for any matrix $A\in \GL$, $A^D \in \GL$ if and only if $A\in H(n)$.
The automaton $\F_D(\A)$ will have the following properties:
\begin{enumerate}[(1)]
\item For every $w\in L(\A)\cap L_{H(n)}$, there exists $w'\in L(\F_D(\A))$ such that $\phi(w')=\phi(w)^D$.
\item For every $w'\in L(\F_D(\A))$, there exists $w\in L(\A)\cap L_{H(n)}$ such that $\phi(w)^D=\phi(w')$.
\end{enumerate}
In other words, we will have
\[
\phi(L(\F_D(\A))) = \{\phi(w)^D\ :\ \text{where } w\in L(\A) \text{ and } \phi(w)\in H(n)\}.
\]

{\bf Construction of the automaton $\F_D(\A)$.} Let $\A$ be a finite automaton in alphabet $\Sigma$ and let $D=\begin{bmatrix} m & 0\\ 0 & mn \end{bmatrix}$ be a diagonal matrix in the Smith normal form, where $m,n\neq 0$.

Suppose that $\A$ has the states $q_0,q_1,\dots,q_t$. Recall from the proof of Lemma \ref{lem:regsub} that the automaton $\A_{H(n)}$, which recognizes $L_{H(n)}$, has the states $u_0,u_1,\dots,u_k$, where $u_0$ is the only initial and also the only final state. First, we construct an automaton $\A'$ for the language $L(\A)\cap L_{H(n)}$ by taking the direct product of $\A$ and $\A_{H(n)}$. Namely, $\A'$ has the states $(q_i,u_j)$, for $i=0,\dots,t$ and $j=0,\dots,k$. The initial states of $\A'$ are of the form $(q_i,u_0)$, where $q_i$ is an initial state of $\A$, and the final states of $\A'$ are of the form $(q_i,u_0)$, where $q_i$ is a final state of $\A$. Furthermore, there is a transition from $(q_i,u_j)$ to $(q_{i'},u_{j'})$ labelled by $\sigma$ if and only if there are transitions $q_i\xrightarrow{\sigma} q_{i'}$ and $u_j\xrightarrow{\sigma} u_{j'}$ in $\A$ and $\A_{H(n)}$, respectively.

Next we replace every transition in $\A'$ by a new path as follows. Let $(q_{i_1},u_{j_1}) \xrightarrow{\sigma} (q_{i_2},u_{j_2})$ be a transition in $\A'$. So there must be a transition of the form $u_{j_1} \xrightarrow{\sigma} u_{j_2}$ in $\A_{H(n)}$. By construction of $\A_{H(n)}$ as described in Lemma \ref{lem:regsub}, we have $U_{j_1}\phi(\sigma)\in H(n)U_{j_2}$ or, equivalently, $U^{}_{j_1}\phi(\sigma)U^{-1}_{j_2}\in H(n)$, where $U_0,\dots,U_k$ are pairwise nonequivalent representatives of the right cosets of $H(n)$ in $\GL$, such that $U_0=I$. Hence $(U^{}_{j_1}\phi(\sigma)U^{-1}_{j_2})^D$ is a matrix with integer coefficients, that is, it belongs to $\GL$. Let $w=\sigma_1\dots \sigma_s\in \Sigma^*$ be a canonical word\footnote{Actually, we can take $w$ to be any word that represents $(U^{}_{j_1}\phi(\sigma)U^{-1}_{j_2})^D$. The fact that it is canonical is not important for our construction.} such that $\phi(w)=(U^{}_{j_1}\phi(\sigma)U^{-1}_{j_2})^D$. Then we replace the transition $(q_{i_1},u_{j_1}) \xrightarrow{\sigma} (q_{i_2},u_{j_2})$ by a path of the form
\[
(q_{i_1},u_{j_1}) \xrightarrow{\sigma_1} p_1 \xrightarrow{\sigma_2}\ \cdots\ \xrightarrow{\sigma_{s-1}} p_{s-1} \xrightarrow{\sigma_s} (q_{i_2},u_{j_2}),
\]
where $p_1,\dots,p_{s-1}$ are new states added to $\A'$. Let $\F_D(\A)$ be an automaton that we obtain after applying the above procedure to $\A'$.


To prove the first property of $\F_D(\A)$, take any $w=\sigma_1\dots \sigma_s\in L(\A)\cap L_{H(n)}$. Then there must be an accepting run $\rho=(q_{i_0},u_{j_0})(q_{i_1},u_{j_1})\dots (q_{i_s},u_{j_s})$ of $\A'$ on $w$. For every transition $(q_{i_{r-1}},u_{j_{r-1}})\xrightarrow{\sigma_r} (q_{i_r},u_{j_r})$ in the run $\rho$, there is a path in $\F_D(\A)$ from $(q_{i_{r-1}},u_{j_{r-1}})$ to $(q_{i_r},u_{j_r})$ labelled by a word $w_r$ such that $\phi(w_r)=(U^{}_{j_{r-1}}\phi(\sigma_r)U^{-1}_{j_r})^D$, where $U^{}_{j_{r-1}}\phi(\sigma_r)U^{-1}_{j_r}\in H(n)$. If we let $w'=w_1\dots w_s$, then $w'$ is accepted by $\F_D(\A)$. To prove that $\phi(w')=\phi(w)^D$, we first note that since $w\in L_{H(n)}$, the run $u_{j_0}u_{j_1}\dots u_{j_s}$ is an accepting run of $\A_{H(n)}$ on $w$, and in particular $j_0=j_s=0$. Since $U^{}_{j_0}=U^{}_{j_s}=U_0=I$, we can rewrite $\phi(w)$ as
\[
\begin{split}
\phi(w) &= U^{^-1}_{j_0}(U^{}_{j_0}\phi(\sigma_1)U_{j_1}^{-1})(U^{}_{j_1}\phi(\sigma_2)U_{j_2}^{-1})\cdots (U^{}_{j_{s-1}}\phi(\sigma_s)U_{j_s}^{-1})U^{}_{j_s}\\
&= (U^{}_{j_0}\phi(\sigma_1)U_{j_1}^{-1})(U^{}_{j_1}\phi(\sigma_2)U_{j_2}^{-1})\cdots (U^{}_{j_{s-1}}\phi(\sigma_s)U_{j_s}^{-1})\ \ (\text{here we used that } U^{}_{j_0}=U^{}_{j_s}=I).
\end{split}
\]
Recall that for each $r=1,\dots,s$, we have $\phi(w_r)=(U^{}_{j_{r-1}}\phi(\sigma_r)U^{-1}_{j_r})^D$. Therefore,
\[
\begin{split}
\phi(w)^D &= (U^{}_{j_0}\phi(\sigma_1)U_{j_1}^{-1})^D(U^{}_{j_1}\phi(\sigma_2)U_{j_2}^{-1})^D\cdots (U^{}_{j_{s-1}}\phi(\sigma_s)U_{j_s}^{-1})^D\\
&= \phi(w_1)\phi(w_2) \cdots \phi(w_s) = \phi(w').
\end{split}
\]
This proves the first property of $\F_D(\A)$.

To prove the second property of $\F_D(\A)$, take any $w'\in L(\F_D(\A))$ and consider an accepting run of $\F_D(\A)$ on $w'$. This run passes through some states of the form $(q_i,u_j)$, that are present in both $\F_D(\A)$ and $\A'$, and some new states that exist only in $\F_D(\A)$. Let $(q_{i_0},u_{j_0}),(q_{i_1},u_{j_1}),\dots,(q_{i_s},u_{j_s})$ be the subsequence of the states of the first type which appear in the accepting run of $\F_D(\A)$. They naturally divide $w'$ into subwords $w'=w_1w_2\dots w_s$, where $w_r$ is a label of the path from $(q_{i_{r-1}},u_{j_{r-1}})$ to $(q_{i_r},u_{j_r})$ for $r=1,\dots,s$. By construction of $\F_D(\A)$, for each $r=1,\dots,s$, there exists a symbol $\sigma_r\in \Sigma$ for which there is a transition $(q_{i_{r-1}},u_{j_{r-1}})\xrightarrow{\sigma_r} (q_{i_r},u_{j_r})$ in $\A'$ and, moreover, $U^{}_{j_{r-1}}\phi(\sigma_r)U^{-1}_{j_r}\in H(n)$ and $\phi(w_r) = (U^{}_{j_{r-1}}\phi(\sigma_r)U^{-1}_{j_r})^D$.

Let $w=\sigma_1\sigma_2\dots \sigma_s$, then $q_{i_0}q_{i_1}\dots q_{i_s}$ will be an accepting run of $\A$ on $w$ and $u_{j_0}u_{j_1}\dots u_{j_s}$ will be an accepting run of $\A_{H(n)}$ on $w$. Thus $w\in L(\A)\cap L_{H(n)}$. Furthermore, we have $u_{j_0}=u_{j_s}=u_0$ and hence $U_{j_0}=U_{j_s}=I$. So we can rewrite $\phi(w)$ as
\[
\begin{split}
\phi(w) &= U^{^-1}_{j_0}(U^{}_{j_0}\phi(\sigma_1)U_{j_1}^{-1})(U^{}_{j_1}\phi(\sigma_2)U_{j_2}^{-1})\cdots (U^{}_{j_{s-1}}\phi(\sigma_s)U_{j_s}^{-1})U^{}_{j_s}\\
&= (U^{}_{j_0}\phi(\sigma_1)U_{j_1}^{-1})(U^{}_{j_1}\phi(\sigma_2)U_{j_2}^{-1})\cdots (U^{}_{j_{s-1}}\phi(\sigma_s)U_{j_s}^{-1}).
\end{split}
\]
From this we obtain the following equalities
\[
\begin{split}
\phi(w)^D &= (U^{}_{j_0}\phi(\sigma_1)U_{j_1}^{-1})^D(U^{}_{j_1}\phi(\sigma_2)U_{j_2}^{-1})^D\cdots (U^{}_{j_{s-1}}\phi(\sigma_s)U_{j_s}^{-1})^D\\ &= \phi(w_1)\phi(w_2) \cdots \phi(w_s) = \phi(w').
\end{split}
\]
This proves the second property of $\F_D(\A)$.

\begin{proposition}\label{prop:1diag}
Let $D$ be a diagonal matrix in the Smith normal form and let $\s_1$ and $\s_2$ be two regular subsets of $\GL$. Then it is decidable whether there exist matrices $A_1\in \s_1$ and $A_2\in \s_2$ such that $A_1DA_2=D$.
\end{proposition}

\begin{proof}
Let $\A_1$ and $\A_2$ be finite automata such that $\s_1=L(\A_1)$ and $\s_2=L(\A_2)$, respectively. We will show that the equation $A_1DA_2=D$ has a solution for some $A_1\in \s_1$ and $A_2\in \s_2$ if and only if $L(\mathrm{Can}(\F_D(\A_1)))\cap L(\mathrm{Can}(\mathrm{Inv}(\A_2))) \neq \emptyset$\footnote{We remind that the construction of the automaton $\mathrm{Can}(\A)$ is described in Section \ref{Can} of the Appendix.}.

First, suppose there exist matrices $A_1\in \s_1$ and $A_2\in \s_2$ such that $A_1DA_2=D$. Let $w_1\in L(\A_1)$ and $w_2\in L(\A_2)$ be such that $\phi(w_1)=A_1$ and $\phi(w_2)=A_2$, respectively. Also let $D=\begin{bmatrix} m & 0\\ 0 & mn\end{bmatrix}$ for some $m,n\neq 0$. We can rewrite the equation $A_1DA_2=D$ as $A_2^{-1}=A_1^D$. From this we can see that the matrix $A_1^D$ must have integer coefficients. Hence, by Proposition \ref{prop:conj}, $A_1\in H(n)$ and $w_1\in L_{H(n)}$. Since $w_1\in L(\A_1)\cap L_{H(n)}$, there exists $w'_1\in L(\F_D(\A_1))$ such that $\phi(w'_1)=\phi(w_1)^D=A_1^D$. Also there is $w'_2\in L(\mathrm{Inv}(\A_2))$ such that $\phi(w'_2)=\phi(w_2)^{-1}=A_2^{-1}$. Since $A_2^{-1}=A_1^D$, we have $\phi(w'_1)=\phi(w'_2)$. In other words, $w'_1$ and $w'_2$ are equivalent. Let $w$ be a canonical word such that $w\sim w'_1\sim w'_2$, then $w\in L(\mathrm{Can}(\F_D(\A_1)))\cap L(\mathrm{Can}(\mathrm{Inv}(\A_2)))$.

Now suppose there is a word $w$ that belongs to $L(\mathrm{Can}(\F_D(\A_1)))\cap L(\mathrm{Can}(\mathrm{Inv}(\A_2)))$. Hence there are words $w'_1$ and $w'_2$ such that $w\sim w'_1\sim w'_2$ and $w'_1\in L(\F_D(\A_1))$ and $w'_2\in L(\mathrm{Inv}(\A_2))$. Therefore, there exists $w_1\in L(\A_1)\cap L_{H(n)}$ such that $\phi(w_1)^D=\phi(w'_1)$. Also there exists $w_2\in L(\A_2)$ such that $\phi(w_2)^{-1} = \phi(w'_2)$. Let $A_1=\phi(w_1)$ and $A_2=\phi(w_2)$. Then we have $A_1^D=\phi(w_1)^D=\phi(w'_1)=\phi(w'_2)=\phi(w_2)^{-1}=A_2^{-1}$, which is equivalent to $A_1DA_2=D$. Moreover, since $w_1\in L(\A_1)$ and $w_2\in L(\A_2)$, we have that $A_1\in \s_1$ and $A_2\in \s_2$.

The proof of the proposition now follows from the facts that the intersection of two regular languages is regular and that the emptiness problem for regular languages is decidable.
\end{proof}

\begin{corollary}\label{cor:1diag}
Let $M_1$ and $M_2$ be nonsingular matrices from $\Z^{2\times 2}$ and let $\s_1$ and $\s_2$ be regular subsets of $\GL$. Then it is decidable whether there exist matrices $A_1\in \s_1$ and $A_2\in \s_2$ such that $A_1M_1A_2=M_2$.
\end{corollary}

\begin{proof}
Let $D_1$ and $D_2$ be the Smith normal forms of $M_1$ and $M_2$, respectively, that is, $M_1=E_1D_1F_1$ and $M_2=E_2D_2F_2$ for some $E_1,F_1,E_2,F_2\in \GL$. Without loss of generality, we can assume that $D_1$ and $D_2$ have strictly positive diagonal coefficients. Note that if the equation $A_1M_1A_2=M_2$ has a solution for some $A_1,A_2\in \GL$, then, by Theorem \ref{SNF}, $M_1$ and $M_2$ must have the same Smith normal form. Therefore, if $D_1\neq D_2$, then the equation does not have a solution.

So suppose that $D=D_1=D_2$ is the Smith normal form of $M_1$ and $M_2$. Then $A_1M_1A_2=M_2$ is equivalent to $A_1(E_1DF_1)A_2=E_2DF_2$, which we can rewrite as $(E_2^{-1}A_1E_1)D(F_1A_2F_2^{-1})=D$. Let $\s'_1=\{E_2^{-1}AE_1 : A\in \s_1\}$ and $\s'_2=\{F_1AF_2^{-1} : A\in \s_2\}$. Then $\s'_1$ and $\s'_2$ are regular subsets of $\GL$ because $E_1,F_1,E_2$, and $F_2$ are some fixed matrices. Now it is not hard to see that the equation $A_1M_1A_2=M_2$ has a solution $A_1,A_2$ such that $A_1\in \s_1$ and $A_1\in \s_2$ if and only if the equation $A'_1DA'_2=D$ has a solution $A'_1,A'_2$ such that $A'_1\in \s'_1$ and $A'_2\in \s'_2$. By Proposition \ref{prop:1diag}, this problem is decidable.
\end{proof}

\subsection{General case: $A_1M_1\dots A_{t-1}M_{t-1}A_t= M_t$}\label{Gen}

To prove an analog of Corollary \ref{cor:1diag} in the general case, we will extend the construction of the automaton $\F_D(\A)$ to  build an automaton $\F(\A_1,\dots,\A_{t-1},M_1,\dots,M_{t-1};M_t)$ (where $\A_1,\dots,\A_{t-1}$ are finite automata in alphabet $\Sigma$ and $M_1,\dots,M_{t-1},M_t$ are nonsingular matrices from $\Z^{2\times 2}$) which will have the following properties:
\begin{enumerate}[(1)]
\item If $w_1\in L(\A_1),\dots,w_{t-1}\in L(\A_{t-1})$ and there is a matrix $A\in \GL$ which satisfies the equation $\phi(w_1)M_1\dots\phi(w_{t-1})M_{t-1}A=M_t$, then there is $w\in L(\F(\A_1,\dots,\A_{t-1},M_1,\dots,M_{t-1};M_t))$ such that $\phi(w_1)M_1\dots\phi(w_{t-1})M_{t-1}\phi(w)^{-1}=M_t$ (and hence $A=\phi(w)^{-1}$).

\item If $w\in L(\F(\A_1,\dots,\A_{t-1},M_1,\dots,M_{t-1};M_t))$, then there are $w_1\in L(\A_1),\dots,w_{t-1}\in L(\A_{t-1})$ such that $\phi(w_1)M_1\dots\phi(w_{t-1})M_{t-1}\phi(w)^{-1}=M_t$.
\end{enumerate}

{\bf Construction of $\F(\A_1,\dots,\A_{t-1},M_1,\dots,M_{t-1};M_t)$.}
The construction will be done by induction on~$t$. We will use the following notations: If $\A_1$ and $\A_2$ are finite automata in alphabet $\Sigma$, then $\A_1\cdot \A_2$ denotes the concatenation of $\A_1$ and $\A_2$. If $\A$ is an automaton and $w\in \Sigma^*$, then $\A\cdot w$ denotes an automaton that recognizes the language $L(\A)\cdot \{w\} = \{uw: u\in L(\A)\}$. Similarly, $w\cdot \A$ is an automaton that recognizes $\{w\}\cdot L(\A) = \{wu: u\in L(\A)\}$.

First, we construct an automaton $\F(\A_1,M_1;M_2)$, which will serve as a base for induction. Let $D_1$ and $D_2$ be  diagonal matrices with nonnegative coefficients which are equal to the Smith normal forms of $M_1$ and $M_2$, respectively. If $D_1\neq D_2$, then define $\F(\A_1,M_1;M_2)$ to be an automaton that accepts the empty language. Otherwise, let $D=D_1=D_2$ be the common Smith normal form of $M_1$ and $M_2$, and suppose $M_1=E_1DF_1$ and $M_2=E_2DF_2$ for some matrices $E_1,F_1,E_2,F_2\in \GL$. Let $w(E_1)$, $w(F_1)$, $w(E_2^{-1})$ and $w(F_2^{-1})$ be canonical words that represent the matrices $E_1$, $F_1$, $E_2^{-1}$ and $F_2^{-1}$, respectively, and define $\F(\A_1,M_1;M_2)$ to be the following automaton
\[
\F(\A_1,M_1;M_2) = w(F_2^{-1})\cdot \F_D\big(w(E_2^{-1})\cdot \A_1\cdot w(E_1)\big)\cdot w(F_1).
\]
The following proposition states that the automaton $\F(\A_1,M_1;M_2)$ indeed satisfies the desired properties.

\begin{proposition} \label{prop:aut1}
Let $\A_1$ be a finite automaton in alphabet $\Sigma$, and let $M_1$ and $M_2$ be nonsingular matrices from $\Z^{2\times 2}$. Then the automaton $\F(\A_1,M_1;M_2)$ has the following properties:
\begin{enumerate}[(1)]
\item If $w_1\in L(\A_1)$ and there is a matrix $A\in \GL$ which satisfies the equation $\phi(w_1)M_1A=M_2$, then there is $w\in L(\F(\A_1,M_1;M_2))$ such that $\phi(w_1)M_1\phi(w)^{-1}=M_2$ (and hence $A=\phi(w)^{-1}$).

\item If $w\in L(\F(\A_1,M_1;M_2))$, then there is $w_1\in L(\A_1)$ such that $\phi(w_1)M_1\phi(w)^{-1}=M_2$.
\end{enumerate}
\end{proposition}

\begin{proof}
Note that if $M_1$ and $M_2$ have different Smith normal forms, then by the uniqueness part of Theorem~\ref{SNF} the equation $A_1M_1A_2=M_2$ cannot have a solution $A_1,A_2\in \GL$. Therefore, in this case both properties of $\F(\A_1,M_1;M_2)$ are trivially satisfied. Now suppose that $D=\begin{bmatrix} m & 0\\ 0 & mn \end{bmatrix}$ is the common Smith normal form of $M_1$ and $M_2$ and let $E_1,F_1,E_2,F_2$ be matrices form $\GL$ such that $M_1=E_1DF_1$ and $M_2=E_2DF_2$.

To see that the first property of $\F(\A_1,M_1;M_2)$ holds, let's take any $w_1\in L(\A_1)$ for which there is a matrix $A\in \GL$ that satisfies the equation $\phi(w_1)M_1A=M_2$. Hence we have that $\phi(w_1)E_1DF_1A=E_2DF_2$, which is equivalent to $F_2^{-1}(E_2^{-1}\phi(w_1)E_1)^D F_1=A^{-1}$. Because $F_2^{-1}$, $F_1$, and $A^{-1}$ are matrices from $\GL$, we conclude that $(E_2^{-1}\phi(w_1)E_1)^D$ is in $\GL$. Then, by Proposition \ref{prop:conj}, we have $E_2^{-1}\phi(w_1)E_1\in H(n)$ or, equivalently, $w(E_2^{-1})\cdot w_1 \cdot w(E_1)\in L_{H(n)}$. By the first property of the construction $\F_D$, there exists $w'\in L\big(\F_D\big(w(E_2^{-1})\cdot \A_1\cdot w(E_1)\big)\big)$ such that $\phi(w')=\phi\big(w(E_2^{-1})\cdot w_1 \cdot w(E_1)\big)^D = (E_2^{-1}\phi(w_1) E_1)^D$. Let $w=w(F_2^{-1})\cdot w'\cdot w(F_1)$. Then $w$ is in $L(\F(\A_1,M_1;M_2))$. Moreover, $\phi(w)= F_2^{-1}\phi(w') F_1 = F_2^{-1}(E_2^{-1}\phi(w_1) E_1)^D F_1$. The last equation is equivalent to $\phi(w_1)E_1DF_1 \phi(w)^{-1}= E_2DF_2$, which is the same as $\phi(w_1)M_1 \phi(w)^{-1}= M_2$. Hence the first property holds.

Now we prove the second property of $\F(\A_1,M_1;M_2)$. Let's take any $w\in L(\F(\A_1,M_1;M_2))$. Then there exists $w'\in L\big(\F_D\big(w(E_2^{-1})\cdot \A_1\cdot w(E_1)\big)\big)$ such that $w=w(F_2^{-1})\cdot w'\cdot w(F_1)$. By the second property of the construction $\F_D$, there exists $w_1\in L(\A_1)$ such that $w(E_2^{-1})\cdot w_1\cdot w(E_1)\in L_{H(n)}$ and $\phi(w') = \phi\big(w(E_2^{-1})\cdot w_1\cdot w(E_1)\big)^D$. The last two conditions are equivalent to the facts that $E_2^{-1}\phi(w_1)E_1\in H(n)$ and $\phi(w') = (E_2^{-1}\phi(w_1)E_1)^D$. From the equation $w=w(F_2^{-1})\cdot w'\cdot w(F_1)$ we have that $\phi(w)=F_2^{-1}\phi(w')F_1$. Therefore, $\phi(w)=F_2^{-1}(E_2^{-1}\phi(w_1)E_1)^DF_1$. The last equation is equivalent to $\phi(w_1)E_1DF_1 \phi(w)^{-1}= E_2DF_2$, which is the same as $\phi(w_1)M_1 \phi(w)^{-1}= M_2$. This proves the second property.

\end{proof}

We now explain how to construct an automaton $\F(\A_1,\dots,\A_{t-1},M_1,\dots,M_{t-1};M_t)$. For convenience the description of this construction is enclosed in the following proposition.

\begin{proposition}\label{prop:autom}
Let $\A_1,\dots,\A_{t-1}$ be finite automata in alphabet $\Sigma$, and let $M_1,\dots,M_{t-1},M_t$ be nonsingular matrices from $\Z^{2\times 2}$. Then there is an automaton $\F(\A_1,\dots,\A_{t-1},M_1,\dots,M_{t-1};M_t)$ which has the following properties:
\begin{enumerate}[(1)]
\item If $w_1\in L(\A_1),\dots,w_{t-1}\in L(\A_{t-1})$ and there is a matrix $A\in \GL$ which satisfies the equation $\phi(w_1)M_1\dots\phi(w_{t-1})M_{t-1}A=M_t$, then there is $w\in L(\F(\A_1,\dots,\A_{t-1},M_1,\dots,M_{t-1};M_t))$ such that $\phi(w_1)M_1\dots\phi(w_{t-1})M_{t-1}\phi(w)^{-1}=M_t$ (and hence $A=\phi(w)^{-1}$).

\item If $w\in L(\F(\A_1,\dots,\A_{t-1},M_1,\dots,M_{t-1};M_t))$, then there are $w_1\in L(\A_1),\dots,w_{t-1}\in L(\A_{t-1})$ such that $\phi(w_1)M_1\dots\phi(w_{t-1})M_{t-1}\phi(w)^{-1}=M_t$.
\end{enumerate}
\end{proposition}

The following lemma will play an important role in the proof of the inductive step in Proposition~\ref{prop:autom}. Informally speaking, it states that when we consider all possible Smith normal forms $\mathit{UDV}$ for a fixed~$D$, we can assume that $U$ comes from a finite set of matrices.

\begin{lemma} \label{lem:fin}
Let $D=\begin{bmatrix} m & 0\\ 0 & mn \end{bmatrix}$ be a diagonal matrix in the Smith normal form and let $U_0,\dots,U_k$ be representatives of the right cosets of $H(n)$ in $\GL$. Then
\[
\{ \mathit{UDV}\ :\ U,V\in \GL\} = \bigcup_{i=0}^k\ \{U_iDV\ :\ V\in \GL\}.
\]
\end{lemma}

\begin{proof}
Consider a matrix $M=\mathit{UDV}$ for some $U,V\in \GL$ and choose $i$ such that $U\in U_iH(n)$. In this case we have that $U_i^{-1}U\in H(n)$, and thus $(U_i^{-1}U)^D$ belongs to $\GL$ by Proposition~\ref{prop:conj}. Let $V'= (U_i^{-1}U)^DV\in \GL$. Then we have an equality $M=\mathit{UDV}=U_iDV'$, and hence $M\in \{U_iDV\ :\ V\in \GL\}$. The inclusion in the other direction is obvious.

\end{proof}

\begin{proof}[Proof of Proposition \ref{prop:autom}]
The proof will be done by induction of $t$. The base case when $t=2$ follows from Proposition \ref{prop:aut1}. Now suppose the proposition holds for $t-1$, and thus we have a construction for the automata of the form $\F(\A_1,\dots,\A_{t-2},M_1,\dots,M_{t-2};M_{t-1})$ which satisfy the properties (1) and (2) above. Using these automata, we will show how to construct an automaton $\F(\A_1,\dots,\A_{t-1},M_1,\dots,M_{t-1};M_t)$.

Let $D_{t-1}=\begin{bmatrix} m & 0\\ 0 & mn \end{bmatrix}$ be equal to the Smith normal form of the matrix $M_{t-1}$ and let $U_0,\dots,U_k$ be representatives of the right cosets of $H(n)$, which can be computed by Theorem \ref{thm:ind}.
Then we define $\F(\A_1,\dots,\A_{t-1},M_1,\dots,M_{t-1};M_t)$ to be an automaton that recognizes the following union of regular languages
\[
\bigcup_{i=0}^k\ L\Big(\F(\A_1,\dots,\A_{t-3},\,\A_{t-2},\,M_1,\dots,M_{t-3},\,M_{t-2}U_iD_{t-1};\,M_t)\cdot \F(\A_{t-1},M_{t-1};U_iD_{t-1})\Big).
\]

To see that the first property holds for $\F(\A_1,\dots,\A_{t-1},M_1,\dots,M_{t-1};M_t)$, let's take $w_1\in L(\A_1),\dots,w_{t-1}\in L(\A_{t-1})$, and suppose there is a matrix $A\in \GL$ which satisfies the equation
\[
\phi(w_1)M_1\dots\phi(w_{t-1})M_{t-1}A=M_t.
\]
By Lemma \ref{lem:fin}, there is $i\in \{0,\dots,k\}$ and $V\in \GL$ such that $\phi(w_{t-1})M_{t-1}A=U_iD_{t-1}V$. So the above equation is equivalent to the following system of equations
\[
\begin{split}
\phi(w_1)M_1\dots\phi(w_{t-2})M_{t-2}U_iD_{t-1}V &=M_t,\\
\phi(w_{t-1})M_{t-1}AV^{-1} &=U_iD_{t-1}.
\end{split}
\]
Since $V\in \GL$, by the inductive hypothesis there is a word $u$ such that
\[
u\in L\Big(\F(\A_1,\dots,\A_{t-3},\,\A_{t-2},\,M_1,\dots,M_{t-3},\,M_{t-2}U_iD_{t-1};\,M_t)\Big)
\]
and
\[
\phi(w_1)M_1\dots\phi(w_{t-2})M_{t-2}U_iD_{t-1}\phi(u)^{-1} =M_t.
\]
Moreover, since $AV^{-1}\in \GL$, by Proposition \ref{prop:aut1}, there is a word $v\in L(\F(\A_{t-1},M_{t-1};U_iD_{t-1}))$ such that $\phi(w_{t-1})M_{t-1}\phi(v)^{-1} =U_iD_{t-1}$. Combining the last two equations together we obtain that
\[
\phi(w_1)M_1\dots\phi(w_{t-1})M_{t-1}\phi(v)^{-1}\phi(u)^{-1}=M_t
\]
or, equivalently,
\[
\phi(w_1)M_1\dots\phi(w_{t-1})M_{t-1}\phi(uv)^{-1}=M_t.
\]
Note that
\[
uv\in L\Big(\F(\A_1,\dots,\A_{t-3},\,\A_{t-2},\,M_1,\dots,M_{t-3},\,M_{t-2}U_iD_{t-1};\,M_t)\cdot \F(\A_{t-1},M_{t-1};U_iD_{t-1})\Big)
\]
and hence $uv\in L(\F(\A_1,\dots,\A_{t-1},M_1,\dots,M_{t-1};M_t))$. Therefore, property (1) holds.

To show the second property, let's take $w\in L(\F(\A_1,\dots,\A_{t-1},M_1,\dots,M_{t-1};M_t))$. Then there is $i\in \{0,\dots,k\}$ such that
\[
w\in L\Big(\F(\A_1,\dots,\A_{t-3},\,\A_{t-2},\,M_1,\dots,M_{t-3},\,M_{t-2}U_iD_{t-1};\,M_t)\cdot \F(\A_{t-1},M_{t-1};U_iD_{t-1})\Big).
\]
Therefore, there are words $u$ and $v$ such that
\[
u\in L\Big(\F(\A_1,\dots,\A_{t-3},\,\A_{t-2},\,M_1,\dots,M_{t-3},\,M_{t-2}U_iD_{t-1};\,M_t)\Big).
\]
and $v\in L(\F(\A_{t-1},M_{t-1};U_iD_{t-1}))$. By Proposition \ref{prop:aut1}, there is $w_{t-1}\in L(\A_{t-1})$ such that
\[
\phi(w_{t-1})M_{t-1}\phi(v)^{-1} =U_iD_{t-1}.
\]
Furthermore, by the inductive hypothesis, there are $w_1\in L(\A_1),\dots,w_{t-2}\in L(\A_{t-2})$ such that
\[
\phi(w_1)M_1\dots\phi(w_{t-2})M_{t-2}U_iD_{t-1}\phi(u)^{-1} =M_t.
\]
Combining the last two equation together we obtain
\[
\phi(w_1)M_1\dots\phi(w_{t-1})M_{t-1}\phi(v)^{-1}\phi(u)^{-1}=M_t.
\]
Note that $\phi(w)^{-1}=\phi(v)^{-1}\phi(u)^{-1}$, and hence we have $\phi(w_1)M_1\dots\phi(w_{t-1})M_{t-1}\phi(w)^{-1}=M_t$.  Therefore, property (2) holds.

\end{proof}

\begin{theorem}\label{thm:diag}
Let $M_1,\dots,M_t$ be nonsingular matrices from $\Z^{2\times 2}$ and let $\s_1,\dots,\s_t$ be regular subsets of $\GL$. Then it is decidable whether there exist matrices $A_1\in \s_1,\dots,A_t\in \s_t$ such that $A_1M_1\dots A_{t-1}M_{t-1}A_t= M_t$.
\end{theorem}

\begin{proof}
Let $\A_1,\dots,\A_t$ be finite automata such that $\s_i=\phi(L(\A_i))$, for each $i=1,\dots,t$. Now consider an automaton $\F(\A_1,\dots,\A_{t-1},M_1,\dots,M_{t-1};M_t)$ which was constructed in the proof of Proposition \ref{prop:autom}. We will show the following equivalence: there exist matrices $A_1\in \s_1,\dots,A_t\in \s_t$ that satisfy the equation $A_1M_1\dots A_{t-1}M_{t-1}A_t= M_t$ if and only if
\[
L\big(\mathrm{Can}(\mathrm{Inv}(\A_t))\big)\ \cap\ L\big(\mathrm{Can}(\F(\A_1,\dots,\A_{t-1},M_1,\dots,M_{t-1};M_t))\big) \neq \emptyset.
\]
The statement of the theorem then follows from the decidability of the emptiness problem for regular languages. 

First, suppose there are matrices $A_1\in \s_1,\dots,A_t\in \s_t$ such that $A_1M_1\dots A_{t-1}M_{t-1}A_t= M_t$. Then there are words $w_1\in L(\A_1),\dots,w_t\in L(\A_t)$ such that
\[
\phi(w_1)M_1\dots \phi(w_{t-1})M_{t-1}\phi(w_t)=M_t.
\]
By property (1) of Proposition \ref{prop:autom}, there is a word $u\in L(\F(\A_1,\dots,\A_{t-1},M_1,\dots,M_{t-1};M_t))$ such that
\[
\phi(w_1)M_1\dots \phi(w_{t-1})M_{t-1}\phi(u)^{-1}=M_t.
\]
In particular, we have $\phi(w_t)=\phi(u)^{-1}$. Furthermore, by the construction of $\mathrm{Inv}(\A_t)$, there is a word $v\in L(\mathrm{Inv}(\A_t))$ such that $\phi(v)=\phi(w_t)^{-1}$. So we have $\phi(u)=\phi(w_t)^{-1}=\phi(v)$, that is, $u\sim v$. Let $w$ be the canonical word that is equivalent to $u$ and $v$. Then
\[
w\in L\big(\mathrm{Can}(\mathrm{Inv}(\A_t))\big)\ \cap\ L\big(\mathrm{Can}(\F(\A_1,\dots,\A_{t-1},M_1,\dots,M_{t-1};M_t))\big).
\]

On the other hand, suppose there is a word $w$ such that
\[
w\in L\big(\mathrm{Can}(\mathrm{Inv}(\A_t))\big)\ \cap\ L\big(\mathrm{Can}(\F(\A_1,\dots,\A_{t-1},M_1,\dots,M_{t-1};M_t))\big).
\]
Then there are words $u$ and $v$ such that $u\sim v\sim w$ and $u\in L(\F(\A_1,\dots,\A_{t-1},M_1,\dots,M_{t-1};M_t))$ and $v\in L(\mathrm{Inv}(\A_t))$. Hence there is $w_t\in L(\A_t)$ such that $\phi(w_t)=\phi(v)^{-1}$. Also by property (2) of Proposition \ref{prop:autom}, there are words $w_1\in L(\A_1),\dots,w_{t-1}\in L(\A_{t-1})$ such that
\[
\phi(w_1)M_1\dots \phi(w_{t-1})M_{t-1}\phi(u)^{-1}=M_t.
\]
Since $v\sim u$, we have that $\phi(u)^{-1}=\phi(v)^{-1}=\phi(w_t)$. Therefore, the above equation is equivalent to
\[
\phi(w_1)M_1\dots \phi(w_{t-1})M_{t-1}\phi(w_t)=M_t.
\]
Now if we let $A_1=\phi(w_1),\dots,A_t=\phi(w_t)$, then for each $i=1,\dots,t$ the matrix $A_i$ belongs to $\s_i$, and hence we have $A_1M_1\dots A_{t-1}M_{t-1}A_t= M_t$.

\end{proof}

\section{Appendix}

\subsection{Construction of the automaton $\mathrm{Can}(\A)$}\label{Can}
Let $\A$ be a finite automaton with alphabet $\Sigma$. We will construct a new automaton $\mathrm{Can}(\A)$ such that the language of $\mathrm{Can}(\A)$ contains only canonical words and $\mathrm{Can}(\A)\sim \A$, that is, $\phi(L(\mathrm{Can}(\A)))=\phi(L(\A))$. In order to do this, we will define a sequence of transformations called $\mathrm{Red}$, $F_N$ and $F_X$ which will have the following properties:
\begin{itemize}
\item $\mathrm{Can}(\A) = F_X\circ \mathrm{Red} \circ F_N(\A)$,
\item $L(F_N(\A)) \sub {\{X,S,R\}}^*\cup N{\{X,S,R\}}^*$, that is, $F_N(\A)$ accepts only those words that have at most one occurrence of $N$ which may appear only in the first position,
\item $L(\mathrm{Red} \circ F_N(\A))\sub {\{X,S,R\}}^*\cup N{\{X,S,R\}}^*$ and, moreover, $\mathrm{Red} \circ F_N(\A)$ accepts only those words that do not contain subwords of the form $\mathit{XX}$, $SX^\alpha S$ and $RX^{\alpha_1}RX^{\alpha_2}R$ for any $\alpha,\alpha_1,\alpha_2\in \{0,1\}$,
\item $F_X\circ \mathrm{Red} \circ F_N(\A)$ accepts only canonical words,
\item finally, we will have the equivalences $\A\sim F_N(\A)\sim \mathrm{Red} \circ F_N(\A)\sim F_X\circ \mathrm{Red} \circ F_N(\A) = \mathrm{Can}(\A)$.
\end{itemize}
We now describe each of these transformations in detail.

{\bf Transformation $F_N$.} We will make use of the following equivalences which can be easily verified: $X\sim \mathit{NXN}$, $S\sim \mathit{NXSN}$, and $R\sim \mathit{NS}R^2\mathit{SN}$.

First, for every transition $q\xrightarrow{X} q'$ which appears in $\A$, we add new states $p_1$, $p_2$ and a new path of the form $q\xrightarrow{N} p_1\xrightarrow{X} p_2\xrightarrow{N} q'$. Note that since $X\sim \mathit{NXN}$, the addition of such paths produces an equivalent automaton. Similarly, for any transition $q\xrightarrow{S} q'$ in $\A$, we add new states $p_1$, $p_2$, $p_3$ and a path $q\xrightarrow{N} p_1\xrightarrow{X} p_2\xrightarrow{S} p_3\xrightarrow{N} q'$. Finally, for any transition $q\xrightarrow{R} q'$ in $\A$, we add new states $p_1$, $p_2$, $p_3$, $p_4$, $p_5$ and a path $q\xrightarrow{N} p_1\xrightarrow{S} p_2\xrightarrow{R} p_3\xrightarrow{R} p_4\xrightarrow{S} p_5\xrightarrow{N} q'$. Again, the addition of such paths produces an equivalent automaton. Let us call this automaton $\A_1$.

Now for every pair of states $q$, $q'$ in $\A_1$, which are connected by a path labelled with $\mathit{NN}$, we add an $\epsilon$-transition $q\xrightarrow{\epsilon} q'$. We repeat this procedure iteratively until no new $\epsilon$-transitions of this type can be added. Let $\A_2$ be the resulting automaton. Note that since $\mathit{NN}$ is equivalent to the empty word, which represents the identity matrix $I$, the automaton $\A_2$ is equivalent to $\A_1$ and hence to~$\A$.

Let $F_N(\A)$ be an automaton that recognizes the intersection $L(\A_2)\cap ({\{X,S,R\}}^*\cup N{\{X,S,R\}}^*)$. Obviously, the language of $F_N(\A)$ is a subset of ${\{X,S,R\}}^*\cup N{\{X,S,R\}}^*$, so we only need to show that $F_N(\A)\sim \A$. Take any $w_1\in L(F_N(\A))$, then $w_1\in L(\A_2)$ and since $\A_2\sim \A$, there is $w_2\in L(\A)$ such that $w_1\sim w_2$. Next, we need to prove that for any $w_2\in L(\A)$, there is $w_1\in L(F_N(\A))$ such that $w_2\sim w_1$.

Let us take any $w_2\in L(\A)$. To construct the required word $w_1$, we first need to find all occurrences of letter $N$ in $w_2$. For example, suppose that $w_2=u_1Nu_2N\dots u_{n-1}Nu_n$, where each $u_i\in {\{X,S,R\}}^*$. If the number of $N$'s is odd, then in each subword $u_i$ with odd $i$ we replace every occurrence of $X$, $S$, and $R$ with $\mathit{NXN}$, $\mathit{NXSN}$, and $\mathit{NS}R^2\mathit{SN}$, respectively, and leave $u_i$'s with even $i$ unchanged. On the other hand, if the number of $N$'s is even, then we apply such substitution to each $u_i$ with even $i$ and leave $u_i$'s with odd $i$ unchanged. Let $w'$ be the resulting word. Then by construction $w'\sim w_2$ and $w'\in L(\A_1)$. Next, we repeatedly remove all occurrences of the subword $\mathit{NN}$ from $w'$. This will give us a word $w_1\sim w'\sim w_2$ such that $w_1\in L(\A_2)$ and $w_1$ contains at most one letter $N$, which may appear in the first position. Hence $w_1\in L(F_N(\A))$. This idea is illustrated by the following example. Let $w_2=\mathit{SXNRNRSNS}\in L(\A)$, so $w_2$ contains an odd number of $N$'s and hence
\[
\begin{split}
w'&=(\mathit{NXSN})(\mathit{NXN})\mathit{NRN}(\mathit{NS}R^2\mathit{SN})(\mathit{NXSN})\mathit{NS}\\
&=\mathit{NXS}(\mathit{NN})X(\mathit{NN})R(\mathit{NN})SR^2S(\mathit{NN})\mathit{XS}(\mathit{NN})S.
\end{split}
\]
In the above formula parentheses are inserted only to visually separated subwords in $w'$. After removing subwords $\mathit{NN}$ from $w'$ we obtain $w_1=\mathit{NXSXRS}R^2\mathit{SXSS}\in L(F_N(\A))$ such that $w_1\sim w_2$.
The next example illustrates the same idea for an even number of $N$'s. Let $w_2=\mathit{SXNRNRSNSN}\in L(\A)$, then
\[
\begin{split}
w'&=\mathit{SXN}(\mathit{NS}R^2\mathit{SN})\mathit{NRSN}(\mathit{NXSN})N\\
&=\mathit{SX}(\mathit{NN})SR^2S(\mathit{NN})\mathit{RS}(\mathit{NN})\mathit{XS}(\mathit{NN}).
\end{split}
\]
After removing $\mathit{NN}$ from $w'$ we obtain $w_1=\mathit{SXS}R^2\mathit{SRSXS}\in L(F_N(\A))$ such that $w_1\sim w_2$. This completes the proof that $F_N(\A)\sim \A$.

{\bf Transformation $\mathrm{Red}$.} To construct $\mathrm{Red}\circ F_N(\A)$ from $F_N(\A)$ we will make use of the following equivalences $SS\sim X$ and $RRR\sim X$. We will also use the fact that $X$ commutes with $S$, $R$, and $N$, and that $\mathit{XX}$ is equivalent to the empty word.

First, we apply the following procedure to $F_N(\A)$:
\begin{enumerate}[(1)]
\item For any pair of states $q$, $q'$ in $F_N(\A)$ that are connected by a path labelled with $\mathit{XX}$, we add an $\epsilon$-transition $q\xrightarrow{\epsilon} q'$.

\item For any pair of states $q$, $q'$ in $F_N(\A)$ that are connected by a path labelled with $SX^\alpha S$, where $\alpha\in \{0,1\}$ (recall that $X^0$ denotes the empty word), we add a new transition $q\xrightarrow{X^\beta} q'$, where $\beta=1-\alpha$.

\item For any pair of states $q$, $q'$ in $F_N(\A)$ that are connected by a path labelled with $RX^{\alpha_1}RX^{\alpha_2}R$, where $\alpha_1, \alpha_2\in \{0,1\}$, we add a new transition $q\xrightarrow{X^\gamma} q'$, where $\gamma \in \{0,1\}$ is such that $\gamma \equiv \alpha_1+\alpha_2+1 \mod 2$.
\end{enumerate}
We repeat the above steps iteratively until no new transitions can be added.

Let $\A'$ be the resulting automaton. By construction, we have $\A'\sim F_N(\A)$. Let $\LL_\mathrm{Red}$ be the regular language which consists of all words in alphabet $\Sigma$ that do not contain subwords of the form $\mathit{XX}$, $SX^\alpha S$ and $RX^{\alpha_1}RX^{\alpha_2}R$ for any $\alpha,\alpha_1,\alpha_2\in \{0,1\}$. Define $\mathrm{Red}\circ F_N(\A)$ as an automaton that accepts the language $L(\A')\cap \LL_\mathrm{Red}$. It is not hard to see that the language of $\mathrm{Red}\circ F_N(\A)$ is contained in $\LL_\mathrm{Red}\cap ({\{X,S,R\}}^*\cup N{\{X,S,R\}}^*)$.

What is left to show is that $\mathrm{Red}\circ F_N(\A)\sim F_N(\A)$. If $w_1\in L(\mathrm{Red}\circ F_N(\A))$, then $w_1\in L(\A')$, and hence $w_1\sim w_2$ for some $w_2\in L(F_N(\A))$ because $\A'\sim F_N(\A)$. On the other hand, if $w_2\in L(F_N(\A))$, then we can repeatedly remove subwords $\mathit{XX}$ from $w_2$ and replace subwords of the form $SX^\alpha S$ and $RX^{\alpha_1}RX^{\alpha_2}R$, for $\alpha,\alpha_1,\alpha_2\in \{0,1\}$, with $X^\beta$ and $X^\gamma$, respectively, where $\beta=1-\alpha$ and $\gamma \in \{0,1\}$ is such that $\gamma \equiv \alpha_1+\alpha_2+1 \mod 2$. Let $w_1$ be a resulting word that does not contain subwords $\mathit{XX}$, $SX^\alpha S$ and $RX^{\alpha_1}RX^{\alpha_2}R$ for any $\alpha,\alpha_1,\alpha_2\in \{0,1\}$. Then $w_1\sim w_2$ and $w_1\in L(\A')\cap \LL_\mathrm{Red}=L(\mathrm{Red}\circ F_N(\A))$.

{\bf Transformation $F_X$.} The words accepted by $\mathrm{Red}\circ F_N(\A)$ are almost in canonical form with the exception that the letter $X$ may appear in the middle of a word. To get rid of such $X$'s we use a similar idea as in the construction of $F_N(\A)$. Namely, we will use the following equivalences: $S\sim \mathit{XSX}$ and $R\sim \mathit{XRX}$. Note that we will not need the equivalence $N\sim \mathit{XNX}$ because the letter $N$ can appear only at the beginning of a word.

To construct $\mathrm{Can}(\A)=F_X\circ \mathrm{Red}\circ F_N(\A)$ from $\mathrm{Red}\circ F_N(\A)$, we do the following. First, for every transition $q\xrightarrow{S} q'$ which appears in $\mathrm{Red}\circ F_N(\A)$, we add new states $p_1$, $p_2$ and a new path of the form $q\xrightarrow{X} p_1\xrightarrow{S} p_2\xrightarrow{X} q'$. Similarly, for every transition $q\xrightarrow{R} q'$ which appears in $\mathrm{Red}\circ F_N(\A)$, we add new states $p_1$, $p_2$ and a new path of the form $q\xrightarrow{X} p_1\xrightarrow{R} p_2\xrightarrow{X} q'$. After that we iteratively add $\epsilon$ transitions $q\xrightarrow{\epsilon} q'$ for every pair of states $q$, $q'$ that are connected by a path with label $\mathit{XX}$. We do this until no new $\epsilon$-transitions can be added. 

Let $\A'$ be the resulting automaton, which is by construction equivalent to $\mathrm{Red}\circ F_N(\A)$. Let $\LL_\mathrm{Can}$ be the regular language which consists of all canonical words in alphabet $\Sigma$. Define $\mathrm{Can}(\A)=F_X\circ \mathrm{Red}\circ F_N(\A)$ as an automaton that accepts the language $L(\A')\cap \LL_\mathrm{Can}$. Therefore, $\mathrm{Can}(\A)$ accepts only canonical words.

The proof that $\mathrm{Can}(\A)\sim \mathrm{Red}\circ F_N(\A)$ is similar to the proof that $F_N(\A)\sim \A$ given above. If $w_1\in L(\mathrm{Can}(\A))$, then $w_1\in L(\A')$ and hence $w_1\sim w_2$ for some $w_2\in L(\mathrm{Red}\circ F_N(\A))$ because $\A'\sim \mathrm{Red}\circ F_N(\A)$. On the other hand, if $w_2\in L(\mathrm{Red}\circ F_N(\A))$, then to construct $w_1\in L(\mathrm{Can}(\A))$ such that $w_1\sim w_2$ we first find all occurrences of the letter $X$ in $w_2$. For example, let $w_2$ has the form $w_2=Nu_1Xu_2X\dots u_{n-1}Xu_n$ or the form $w_2=u_1Xu_2X\dots u_{n-1}Xu_n$, where each $u_i\in {\{S,R\}}^*$. If the number of $X$'s is odd, then in each $u_i$ with odd $i$ we replace every occurrence of $R$ and $S$ with $\mathit{XRX}$ and $\mathit{XSX}$, respectively, and leave $u_i$'s with even $i$ unchanged. If the number of $X$'s is even, then we do the same substitution in all $u_i$'s with even $i$ and leave $u_i$'s with odd $i$ unchanged. After that we remove all occurrences of $\mathit{XX}$. If $w_1$ is a resulting word, then $w_1\sim w_2$ and $w_1\in L(\A')$. Moreover, since $w_1$ is in canonical form, we also have $w_1\in L(\mathrm{Can}(\A))$. This idea is illustrated by the following example. Suppose $w_2=\mathit{NSRXSXRRX}$, then after replacing suitable occurrences of $R$ and $S$ with $\mathit{XRX}$ and $\mathit{XSX}$, respectively, we obtain the word
\[
\begin{split}
&N(\mathit{XSX})(\mathit{XRX})\mathit{XSX}(\mathit{XRX})(\mathit{XRX})X\\
= & \mathit{NXS}(\mathit{XX})R(\mathit{XX})S(\mathit{XX})R(\mathit{XX})R(\mathit{XX}).
\end{split}
\]
After removing all occurrences of $\mathit{XX}$ we obtain the word $w_1=\mathit{NXSRSRR}\sim w_2$ which is in canonical form, and hence $w_1\in L(\mathrm{Can}(\A))$. This completes the construction of $\mathrm{Can}(\A)$.

\bibliographystyle{abbrv}
\bibliography{refs}

\begin{thebibliography}{10}

\bibitem{Babai}
L.~Babai, R.~Beals, J.-y. Cai, G.~Ivanyos, and E.~M. Luks.
\newblock Multiplicative equations over commuting matrices.
\newblock In {\em Proceedings of the Seventh Annual ACM-SIAM Symposium on
  Discrete Algorithms}, SODA '96, pages 498--507, Philadelphia, PA, USA, 1996.
  Society for Industrial and Applied Mathematics.

\bibitem{BP2008}
P.~Bell and I.~Potapov.
\newblock On undecidability bounds for matrix decision problems.
\newblock {\em Theoretical Computer Science}, 391(1-2):3--13, 2008.

\bibitem{BP_IC2008}
P.~Bell and I.~Potapov.
\newblock Reachability problems in quaternion matrix and rotation semigroups.
\newblock {\em Information and Computation}, 206(11):1353--1361, 2008.

\bibitem{BHP2012}
P.~C. Bell, M.~Hirvensalo, and I.~Potapov.
\newblock Mortality for 2x2 matrices is {NP}-hard.
\newblock In B.~Rovan, V.~Sassone, and P.~Widmayer, editors, {\em Mathematical
  Foundations of Computer Science 2012}, volume 7464 of {\em Lecture Notes in
  Computer Science}, pages 148--159. Springer Berlin Heidelberg, 2012.

\bibitem{Identity}
P.~C. Bell and I.~Potapov.
\newblock On the undecidability of the identity correspondence problem and its
  applications for word and matrix semigroups.
\newblock {\em Int. J. Found. Comput. Sci.}, 21(6):963--978, 2010.

\bibitem{BP2012}
P.~C. Bell and I.~Potapov.
\newblock On the computational complexity of matrix semigroup problems.
\newblock {\em Fundam. Inf.}, 116(1-4):1--13, Jan. 2012.

\bibitem{Blondel2005}
V.~D. Blondel, E.~Jeandel, P.~Koiran, and N.~Portier.
\newblock Decidable and undecidable problems about quantum automata.
\newblock {\em SIAM J. Comput.}, 34(6):1464--1473, June 2005.

\bibitem{solution10-3}
V.~D. {Blondel} and A.~{Megretski}, editors.
\newblock {\em {Unsolved problems in mathematical systems and control theory.}}
\newblock Princeton, NJ: Princeton University Press, 2004.
\newblock http://press.princeton.edu/math/blondel/solutions.html.

\bibitem{CassaigneHHN14}
J.~Cassaigne, V.~Halava, T.~Harju, and F.~Nicolas.
\newblock Tighter undecidability bounds for matrix mortality,
  zero-in-the-corner problems, and more.
\newblock {\em CoRR}, abs/1404.0644, 2014.

\bibitem{CHK99}
J.~Cassaigne, T.~Harju, and J.~Karhumaki.
\newblock On the undecidability of freeness of matrix semigroups.
\newblock {\em International Journal of Algebra and Computation},
  09(03n04):295--305, 1999.
\newblock http://www.worldscientific.com/doi/pdf/10.1142/S0218196799000199.

\bibitem{CK2005}
C.~Choffrut and J.~Karhumaki.
\newblock Some decision problems on integer matrices.
\newblock {\em RAIRO-Theor. Inf. Appl.}, 39(1):125--131, 2005.

\bibitem{STOC2013}
V.~Chonev, J.~Ouaknine, and J.~Worrell.
\newblock The orbit problem in higher dimensions.
\newblock In {\em Symposium on Theory of Computing Conference, STOC'13, Palo
  Alto, CA, USA, June 1-4, 2013}, pages 941--950, 2013.

\bibitem{COW_JACM}
V.~Chonev, J.~Ouaknine, and J.~Worrell.
\newblock On the complexity of the orbit problem.
\newblock {\em to apear in JACM}, 2016.

\bibitem{EFM1999}
J.~Esparza, A.~Finkel, and R.~Mayr.
\newblock On the verification of broadcast protocols.
\newblock In {\em Logic in Computer Science, 1999. Proceedings. 14th Symposium
  on}, pages 352--359, 1999.

\bibitem{GOW_STACS2015}
E.~Galby, J.~Ouaknine, and J.~Worrell.
\newblock {On Matrix Powering in Low Dimensions}.
\newblock In E.~W. Mayr and N.~Ollinger, editors, {\em 32nd International
  Symposium on Theoretical Aspects of Computer Science (STACS 2015)}, volume~30
  of {\em Leibniz International Proceedings in Informatics (LIPIcs)}, pages
  329--340, Dagstuhl, Germany, 2015. Schloss Dagstuhl--Leibniz-Zentrum fuer
  Informatik.

\bibitem{GB94}
A.~Gerrard and J.~M. Burch.
\newblock {\em Introduction to matrix methods in optics}.
\newblock Dover Publications, Inc., New York, 1994.
\newblock Corrected reprint of the 1975 original.

\bibitem{Gurevich2007}
Y.~Gurevich and P.~Schupp.
\newblock Membership problem for the modular group.
\newblock {\em SIAM J. Comput.}, 37(2):425--459, May 2007.

\bibitem{tHaHaHiKa05a}
V.~Halava, T.~Harju, M.~Hirvensalo, and J.~Karhumaki.
\newblock Skolem's problem - on the border between decidability and
  undecidability.
\newblock Technical Report 683, Turku Centre for Computer Science, 2005.

\bibitem{KB79}
R.~Kannan and A.~Bachem.
\newblock Polynomial algorithms for computing the {S}mith and {H}ermite normal
  forms of an integer matrix.
\newblock {\em SIAM J. Comput.}, 8(4):499--507, 1979.

\bibitem{KL86}
R.~Kannan and R.~J. Lipton.
\newblock Polynomial-time algorithm for the orbit problem.
\newblock {\em J. ACM}, 33(4):808--821, Aug. 1986.

\bibitem{LP2004}
A.~Lisitsa and I.~Potapov.
\newblock Membership and reachability problems for row-monomial
  transformations.
\newblock In {\em Mathematical Foundations of Computer Science 2004, 29th
  International Symposium, {MFCS} 2004, Prague, Czech Republic, August 22-27,
  2004, Proceedings}, pages 623--634, 2004.

\bibitem{LS}
R.~C. Lyndon and P.~E. Schupp.
\newblock {\em Combinatorial group theory}.
\newblock Springer-Verlag, Berlin-New York, 1977.
\newblock Ergebnisse der Mathematik und ihrer Grenzgebiete, Band 89.

\bibitem{MKS}
W.~Magnus, A.~Karrass, and D.~Solitar.
\newblock {\em Combinatorial group theory}.
\newblock Dover Publications, Inc., New York, revised edition, 1976.
\newblock Presentations of groups in terms of generators and relations.

\bibitem{Markov}
A.~Markov.
\newblock On certain insoluble problems concerning matrices.
\newblock {\em Doklady Akad. Nauk SSSR}, 57(6):539--542, June 1947.

\bibitem{OSW2015}
J.~Ouaknine, J.~a.~S. Pinto, and J.~Worrell.
\newblock On termination of integer linear loops.
\newblock In {\em Proceedings of the Twenty-Sixth Annual ACM-SIAM Symposium on
  Discrete Algorithms}, SODA '15, pages 957--969. SIAM, 2015.

\bibitem{OW_ICALP2015-1}
J.~Ouaknine and J.~Worrell.
\newblock On the positivity problem for simple linear recurrence sequences,.
\newblock In {\em Automata, Languages, and Programming - 41st International
  Colloquium, {ICALP} 2014, Copenhagen, Denmark, July 8-11, 2014, Proceedings,
  Part {II}}, pages 318--329, 2014.

\bibitem{OW_ICALP2015-2}
J.~Ouaknine and J.~Worrell.
\newblock Ultimate positivity is decidable for simple linear recurrence
  sequences.
\newblock In {\em Automata, Languages, and Programming - 41st International
  Colloquium, {ICALP} 2014, Copenhagen, Denmark, July 8-11, 2014, Proceedings,
  Part {II}}, pages 330--341, 2014.

\bibitem{PS15}
I.~Potapov and P.~Semukhin.
\newblock Vector reachability problem in {SL}(2, {Z)}.
\newblock {\em CoRR}, abs/1510.03227, 2015.
\newblock http://arxiv.org/abs/1510.03227.

\bibitem{Ran}
R.~A. Rankin.
\newblock {\em Modular forms and functions}.
\newblock Cambridge University Press, Cambridge-New York-Melbourne, 1977.

\end{thebibliography}

\end{document}